\documentclass[a4paper,11pt]{fullverllncs}
\usepackage[left=2.5cm,top=2.5cm,right=2.5cm,centering]{geometry}
\usepackage{graphicx}

\usepackage{amsmath,amssymb}
\usepackage{graphicx}
\usepackage{ascmac}
\usepackage{array}
\usepackage{algpseudocode}
\usepackage{amsfonts}
\usepackage{amssymb}
\usepackage{amsmath}
\usepackage{algorithm, algpseudocode}
\usepackage{multirow}
\usepackage{arydshln}
\usepackage{hhline} 
\usepackage[misc,geometry]{ifsym} 

\usepackage[dvipdfmx, colorlinks]{hyperref, xcolor}
\definecolor{winered}{rgb}{0.5,0,0}
\definecolor{darkblue}{rgb}{0,0,0.5}
\definecolor{darkgreen}{rgb}{0,0.3,0}
\hypersetup{
linkcolor=winered,
citecolor=darkblue,
urlcolor=darkgreen
}
\newcommand{\A}{\mathsf{A}}
\newcommand{\B}{\mathsf{B}}
\newcommand{\C}{\mathsf{C}}

\newcommand{\sfS}{\mathsf{S}}
\newcommand{\Ext}{\mathsf{Ext}}

\newcommand{\Adv}{\mathsf{Adv}}

\newcommand{\N}{\mathbb{N}}

\newcommand{\ttfalse}{\mathtt{false}}
\newcommand{\tttrue}{\mathtt{true}}

\newcommand{\info}{\mathtt{info}}

\newcommand{\sfBasic}{\mathsf{Base}}

\newcommand{\Ours}{\mathsf{Ours}}
\newcommand{\sfGame}{\mathsf{Game}}

\newcommand{\vk}{\mathsf{vk}}
\newcommand{\sk}{\mathsf{sk}}
\newcommand{\pp}{\mathsf{pp}}
\newcommand{\DS}{\mathsf{DS}}
\newcommand{\DSSetup}{\mathsf{DS.Setup}}
\newcommand{\DSKeyGen}{\mathsf{DS.KeyGen}}
\newcommand{\DSSign}{\mathsf{DS.Sign}}
\newcommand{\DSVerify}{\mathsf{DS.Verify}}
\newcommand{\ppDS}{\mathsf{pp}^{\mathsf{DS}}}
\newcommand{\skDS}{\mathsf{sk}^{\mathsf{DS}}}
\newcommand{\vkDS}{\mathsf{vk}^{\mathsf{DS}}}
\newcommand{\sigmaDS}{\sigma^{\mathsf{DS}}}
\newcommand{\sigmastrDS}{\sigma^{*\mathsf{DS}}}
\newcommand{\nDSSetup}{\mathsf{Setup}}
\newcommand{\nDSKeyGen}{\mathsf{KeyGen}}
\newcommand{\nDSSign}{\mathsf{Sign}}
\newcommand{\nDSVerify}{\mathsf{Verify}}

\newcommand{\sEUFCMA}{\mathsf{sEUFCMA}}
\newcommand{\Sign}{\mathsf{Sign}}

\newcommand{\COM}{\mathsf{COM}}
\newcommand{\COMKeyGen}{\mathsf{COM.KeyGen}}
\newcommand{\COMCommit}{\mathsf{COM.Commit}}
\newcommand{\nCOMKeyGen}{\mathsf{KeyGen}}
\newcommand{\nCOMCommit}{\mathsf{Commit}}
\newcommand{\ck}{\mathsf{ck}}
\newcommand{\sBind}{\mathsf{sBind}}
\newcommand{\Hide}{\mathsf{Hide}}

\newcommand{\MerkleTree}{\mathsf{MerkleTree}}
\newcommand{\MerklePath}{\mathsf{MerklePath}}
\newcommand{\RootReconstruct}{\mathsf{RootReconstruct}}
\newcommand{\sftree}{\mathsf{tree}}
\newcommand{\sfroot}{\mathsf{root}}
\newcommand{\sfpath}{\mathsf{path}}

\newcommand{\ItoB}{\mathsf{I2B}}

\newcommand{\Coll}{\mathsf{Coll}}
\newcommand{\MT}{\mathsf{MT}}
\newcommand{\onOS}{(1, n)\mathsf{\mathchar`-OS}}
\newcommand{\OS}{\mathsf{OS}}
\newcommand{\OSSetup}{\mathsf{OS.Setup}}
\newcommand{\OSKeyGen}{\mathsf{OS.KeyGen}}
\newcommand{\OSUi}{\mathsf{OS.U}_{1}}
\newcommand{\OSSii}{\mathsf{OS.S}_{2}}
\newcommand{\OSUDer}{\mathsf{OS.U}_{\mathsf{Der}}}
\newcommand{\OSVerify}{\mathsf{OS.Verify}}
\newcommand{\ppOS}{\mathsf{pp}^{\mathsf{OS}}}
\newcommand{\skOS}{\mathsf{sk}^{\mathsf{OS}}}
\newcommand{\vkOS}{\mathsf{vk}^{\mathsf{OS}}}
\newcommand{\sigmaOS}{\sigma^{\mathsf{OS}}}
\newcommand{\sigmastrOS}{\sigma^{*\mathsf{OS}}}
\newcommand{\nOSSetup}{\mathsf{Setup}}
\newcommand{\nOSKeyGen}{\mathsf{KeyGen}}
\newcommand{\nOSUi}{\mathsf{U}_{1}}
\newcommand{\nOSSii}{\mathsf{S}_{2}}
\newcommand{\nOSUDer}{\mathsf{U}_{\mathsf{Der}}}
\newcommand{\nOSVerify}{\mathsf{Verify}}
\newcommand{\st}{\mathsf{st}}
\newcommand{\Amb}{\mathsf{Amb}}
\newcommand{\SeqsEUFCMA}{\mathsf{Seq\mathchar`- sEUFCMA}}
\newcommand{\Fin}{\mathsf{Fin}}

\newcommand{\ListM}{\mathsf{ListM}}
\newcommand{\ZLH}{\mathsf{ZLH}}

\newcommand{\Final}{\mathsf{Final}}
\newcommand{\Hash}{\mathsf{Hash}}
\newcommand{\Link}{\mathsf{Link}}

\newcommand{\rmsEUFCMA}{\mathrm{sEUF}\mathchar`-\mathrm{CMA}}
\newcommand{\rmSeqsEUFCMA}{\mathrm{Seq}\mathchar`-\mathrm{sEUF}\mathchar`-\mathrm{CMA}}

\newcommand{\ReuseDS}{\mathtt{DS}_{\mathtt{reuse}}}
\newcommand{\ForgeDS}{\mathtt{DS}_{\mathtt{forge}}}
\newcommand{\CollCOM}{\mathtt{COM}_{\mathtt{coll}}}

\spnewtheorem{assumption}{Assumption}{\bfseries}{\itshape}

\begin{document}

\title{$1$-out-of-$n$ Oblivious Signatures:\\ Security Revisited and a Generic Construction \\ with an Efficient Communication Cost\thanks{A preliminary version \cite{TT23} of this paper is appeared in Information Security and Cryptology  {ICISC} 2023  - 26th International Conference.}}
\author{Masayuki Tezuka\inst{1}\textsuperscript{(\Letter)} \and Keisuke Tanaka\inst{2}}
\authorrunning{M.Tezuka et al.}

\institute{Tokyo Institute of Technology, Tokyo, Japan\\
\email{tezuka.m.ac@m.titech.ac.jp} 
}

\maketitle
\pagestyle{plain}
\noindent
\makebox[\linewidth]{March 31, 2024}              

\begin{abstract}
$1$-out-of-$n$ oblivious signature by Chen (ESORIC 1994) is a protocol between the user and the signer.
In this scheme, the user makes a list of $n$ messages and chooses the message that the user wants to obtain a signature  from the list.
The user interacts with the signer by providing this message list and obtains the signature for only the chosen message without letting the signer identify which messages the user chooses.
Tso et al. (ISPEC 2008) presented a formal treatment of $1$-out-of-$n$ oblivious signatures.
They defined unforgeability and ambiguity for $1$-out-of-$n$ oblivious signatures as a security requirement.

In this work, first, we revisit the unforgeability security definition by Tso et al. and point out that their security definition has problems.
We address these problems by modifying their security model and redefining unforgeable security.
Second, we improve the generic construction of a $1$-out-of-$n$ oblivious signature scheme by Zhou et al. (IEICE Trans 2022).
We reduce the communication cost by modifying their scheme with a Merkle tree.
Then we prove the security of our modified scheme.

\end{abstract}

\section{Introduction}
\subsection{Background}

\paragraph{\bf Oblivious Signatures.}
The notion of $1$-out-of-$n$ oblivious signatures by Chen~\cite{Chen94} is an interactive protocol between a signer and a user.
In an oblivious signature scheme, first, the user makes a list of $n$ messages $M=(m_{i})_{i \in \{1, \dots, n\}}$ and chooses one of message $m_{j}$ in $M$ that the user wants to obtain a signature.
Then the user interacts with the signer by sending the list $M$ with a first message $\mu$ at the beginning of the interaction.
The signer can see the candidate messages $M$ that the user wants to get signed, but cannot identify which one of the messages in $M$ is chosen by the user.
After completing the interaction with the signer, the user can obtain a signature $\sigma$ for only the chosen message $m_{j}$.

$1$-out-of-$n$ oblivious signatures should satisfy ambiguity and unforgeability. 
Ambiguity prevents the signer from identifying which one of the messages the signer wants to obtain the signature in the interaction.
Unforgeability requires that for each interaction, the user cannot obtain a signature of a message $m \notin M$ and can obtain a signature for only one message $m \in M$ where $M$ is a list of message that the user sends to the signer at the beginning of the interaction.

Oblivious signatures can be used to protect the privacy of users.
Chen \cite{Chen94} explained an application of oblivious signatures as follows.
The user will buy software from the seller and the signature from the seller is needed to use the software.
However, information about which software the user is interested in may be sensitive at some stage.
In this situation, by using oblivious signatures, the user can make a list of $n$ software and obtain a signature only for the one software that the user honestly wants to obtain without revealing it to the seller (signer).
The oblivious signature can be used for e-voting systems \cite{CC18,SYL08}.

\paragraph{\bf Oblivious Signatures and Blind Signatures.}
Signatures with a similar flavor to oblivious signatures are blind signatures proposed by Chaum \cite{Chaum82}.
In a blind signature scheme, similar to an oblivious signature scheme, a user chooses a message and obtains a corresponding signature by interacting with the signer, but no message candidate list in the blind signature scheme. 
Typically, blind signatures satisfy blindness and one-more unforgeability (OMUF). 
Blindness guarantees that the signer has no idea what a message is being signed and prevents the signer from linking a message/signature pair to the run of the protocol where it was created.
OMUF security prevents the user from forging a new signature.

From the point of view of hiding the contents of the message, it may seem that blind signatures are superior than oblivious signatures.
But compared to blind signatures, oblivious signature has merits listed as follows.
\begin{itemize}
\item {{\bf Avoid Signing Disapprove Messages:}}
In blind signatures, since the signer has no information about the message that the user wants to obtain the signature, the signer cannot prevent users from obtaining a signature on the message that the signer does not want to approve.

Partially blind signatures proposed by Abe and Fujisaki \cite{AF96} mitigate this problem.
This scheme allows the user and the signer to agree on a predetermined piece of common information $\info$ which must be included in the signed message.
However, similar to blind signatures, the signer has no information for the blinded part of a message, partially blind signatures do not provide a full solution for the above problem.

By contrast, oblivious signatures allow the signer to view a list of messages.
If the message that the signer does not want to approve is included in the message list, the signer can refuse to sign.
Thus, the ambiguity of oblivious signatures provides a better solution for the above problem.

\item {{\bf Based on Weaker Assumptions:}}
Recent works on blind signatures are dedicated to constructing efficient round-optimal (i.e., $2$-move signing interaction) blind signature schemes \cite{AKSY22,BHKKKSW19,PK22,FHKS16,FHS15,Ghadafi17,HK16,HK17,HLW23,KNYY21,LNP22}.
However, these schemes either rely on at least one of strong primitives, models, or assumptions such as pairing groups \cite{BHKKKSW19,FHKS16,FHS15,Ghadafi17,HK17,HLW23}, non-interactive zero-knowledge (NIZK) \cite{AKSY22,PK22,KNYY21,LNP22}, the random oracle model (ROM) \cite{PK22,HLW23}, the generic group model (GGM) \cite{FHKS16}, interactive assumptions \cite{BHKKKSW19,FHS15,Ghadafi17,HK17}, $q$-type assumptions~\cite{HK16}, one-more assumptions \cite{AKSY22}, or knowledge assumptions \cite{HK16}.

By contrast, a generic construction of a round-optimal oblivious signature scheme without the ROM was proposed in the recent work by Zhou, Liu, and Han \cite{ZLH22}.
This construction uses a digital signature scheme and a commitment scheme. This leads to instantiations in various standard assumptions (e.g., DDH, DCR, Factoring, RSA, LWE) without the ROM.
Thus, the round-optimal oblivious signature schemes can be constructed with weaker assumptions than round-optimal blind signature schemes.
\end{itemize}

\paragraph{\bf Previous Works on Oblivious Signatures.}
The notion of oblivious signatures was introduced by Chen \cite{Chen94} and proposed $1$-out-of-$n$ oblivious signature schemes in the ROM.
Following this seminal work, several $1$-out-of-$n$ oblivious signature schemes have been proposed.

Tso, Okamoto, and Okamoto \cite{TOO08} formalized the syntax and security definition of the $1$-out-of-$n$ oblivious signature scheme.
They gave the efficient round-optimal (i.e., $2$-move) $1$-out-of-$n$ oblivious signature scheme based on the Schnorr signature scheme.
The security of this scheme can be proven under the DL assumption in the ROM.

Chiou and Chen \cite{CC18} proposed a $t$-out-of-$n$ oblivious signature scheme.
This scheme needs $3$ rounds for a signing interaction and the security of this scheme can be proven under the RSA assumption in the ROM.

You, Liu, Tso, Tseng, and Mambo \cite{YLTTM22} proposed the lattice-based $1$-out-of-$n$ oblivious signature scheme.
This scheme is round-optimal and the security can be proven under the short integer solution (SIS) problem in the ROM.

In recent work by Zhou, Liu, and Han \cite{ZLH22}, a generic construction of a round-optimal $1$-out-of-$n$ oblivious signature scheme was proposed.
Their scheme is constructed from a commitment scheme and a digital signature scheme without the ROM.
By instantiating a signature scheme and commitment scheme from standard assumptions without the ROM, this generic construction leads $1$-out-of-$n$ oblivious signature schemes from standard assumptions without the ROM.
As far as we know, their scheme is the first generic construction of a $1$-out-of-$n$ oblivious signature scheme without the ROM.

\subsection{Motivation}
The security model for a $1$-out-of-$n$ oblivious signature scheme is formalized by Tso \cite{TOO08}.
Their security model is fundamental for subsequent works \cite{ZLH22,YLTTM22}.
However, this security model has several problems.
Here, we briefly review the unforgeability security model in \cite{TOO08} and explain the problems of their model.
The formal description of this security game is given in Section \ref{Subsec_Def_Unf_Prev}.

\paragraph{\bf Definition of Unforgeability in \cite{TOO08}.}
Informally, the unforgeability for a $1$-out-of-$n$ oblivious signature scheme in \cite{TOO08} is defined by the following game.

Let $\A$ be an adversary that executes a user part and tries to forge a new signature.
$\A$ engages in the signing interaction with the signer.
$\A$ can make any message list $M_{i}$ and any one message $m_{i,j_{i}} \in M_{i}$.
Then, $\A$ engages the $i$-th signing interaction by sending $M_{i}$ the signer at the beginning of the interaction.
By interacting with the signer, $\A$ can obtain a signature $\sigma_{i}$ on a message $m_{i,j_{i}}$.
Let $t$ be the number of signing interaction with the signer and $\A$. Let $\mathbb{L}^{\Sign} = \{m_{i,j_{i}}\}_{i \in \{1, \dots, t\}}$ be all messages that $\A$ has obtained signatures.
$\A$ wins this game if $\A$ outputs a valid signature $\sigma^*$ on a message $m^* \notin \mathbb{L}^{\Sign}$.
A $1$-out-of-$n$ oblivious signature scheme satisfies unforgeability if for all PPT adversaries $\A$ cannot win the above game in non-negligible probability. 
However, the above security game has several problems listed below.

\begin{itemize}
\item{\bf Problem 1: How to Store Messages in $\mathbb{L}^{\Sign}$:}
In the above security game, we need to store corresponding messages that $\A$ has obtained signatures.
However, by ambiguity property, we cannot identify the message $m_{i, j_{i}}$ that $\A$ selected to obtain a signature from a transcription of the $i$-th interaction with $M_{i}$.
This problem can be addressed by forcing $\A$ to output $(m_{i, j_{i}}, \sigma_{i})$ at the end of each signing query.
However, the next problem arises.

\item{\bf Problem 2: Trivial Attack:}
One problem is the existence of a trivial attack on the security game.
Let us consider the following adversary $\A$ that runs signing protocol execution twice.
$\A$ chooses $M=(m_{0}, m_{1})$ where $m_{0}$ and $m_{1}$ are distinct, and sets lists as $M_{1}=M_{2} = M$.
In the 1st interaction, $\A$ chooses $m_{0} \in M_{1}$, obtains a signature $\sigma_{0}$ on a message $m_{0}$, and outputs $(m_{0}, \sigma_{0})$ at the end of interaction.
In the 2nd interaction, $\A$ chooses $m_{1} \in M_{2}$, obtains a signature $\sigma_{1}$ on a message $m_{1}$, and outputs $(m_{0}, \sigma_{0})$ at the end of interaction.
Then, $\A$ outputs a trivial forgery $(m^*, \sigma^*) = (m_{1}, \sigma_{1})$.
This problem occurs when $\A$ pretends to have obtained $(m_{0}, \sigma_{0})$ in the 2nd signing interaction.
Since the security model lacks a countermeasure for this trivial attack, $\A$ succeeds in this trivial attack.
The unforgeability security models in previous works \cite{YLTTM22,ZLH22} are based on the model by Tso et al. \cite{TOO08}.
This trivial attack also works for these security models as well.

Note that we only claim that the security model in \cite{TOO08} has problems.
We do not intend to claim that existing constructions in \cite{Chen94,TOO08,YLTTM22,ZLH22} are insecure.
Under an appropriate unforgeability security model, it may be possible to prove the security for these constructions.
Some constructions seem to have structures that can address this trivial attack.
Reassessing the unforgeability security of these constructions in an appropriate model is a future work.

\item{\bf Problem 3: Missing Adversary Strategy:}
The security game does not capture an adversary with the following strategy.
Let us consider an adversary $\A$ that executes the signing query only once.
$\A$ interacts with the signer with a message list $M$ and intends to obtain a signature $\sigma^*$ on a message $m^* \notin M$, but give up outputting $(m, \sigma)$ where $m \in M$ at the end of the signing query.
The game only considers the adversary that outputting $(m, \sigma)$ where $m \in M$ at the end of the signing query.
At the heart of unforgeability security, we should guarantee that $\A$ cannot obtain a signature on a message which is not in the list $M$. Thus, we should consider this adversary strategy.
\end{itemize}

\subsection{Our Contribution}
The first contribution is providing a new security definition of the unforgeability security for a $1$-out-of-$n$ oblivious signature scheme.
We address the problems described in the previous section.
We refer the reader to Section \ref{SubsecNewUnf} for more detail on our unforgeability security definition.

The second contribution is an improvement of a generic construction of $1$-out-of-$n$ oblivious signature schemes by \cite{YLTTM22}.
This round-optimal construction is obtained by a simple combination of a digital signature scheme and a commitment scheme.
However, a bottleneck of this scheme is the communication cost (See Fig. \ref{CompareGenericOS}). 
\begin{figure}[h]
\begin{center}
\begin{tabular}{|c||c|c|c|c|c|}\hline
Scheme & $|\vkOS|$ & $|\mu|$ & $|\rho|$ & $|\sigmaOS|$\\
\hline
\hline
\multirow{2}{*}{
\begin{tabular}{c}
$\OS_{\ZLH}$\\
\cite{ZLH22} 
\end{tabular}
}&
\multirow{2}{*}{$ |\vkDS|$} &
\multirow{2}{*}{$|c^{\COM}|$} &
\multirow{2}{*}{~$n|\sigmaDS|$} & 
\multirow{2}{*}{$|\sigmaDS| + |c^{\COM}| + |r^{\COM}|$}\\
&&&&\\
\hline
\multirow{2}{*}{
\begin{tabular}{c}
$\OS_{\Ours}$\\
\S \ref{SubsecOurGenCon}
\end{tabular}{}
}& 
\multirow{2}{*}{$|\vkDS|$} & 
\multirow{2}{*}{$|c^{\COM}|$} &
\multirow{2}{*}{$|\sigmaDS|$} &
\multirow{2}{*}{$|\sigmaDS| + |c^{\COM}| + |r^{\COM}| + (\lceil \log_{2}n \rceil + 1)\lambda + \lceil \log_{2}n \rceil $}  \\
&&&&\\
\hline
\end{tabular}\\
\end{center}
\caption{\small Comparison with generic construction of $1$-out-of-$n$ oblivious signature schemes.}
$|\vkOS|$ represents the bit length of the verification key, $|\mu|$ represents the bit length of the first communication, $|\rho|$ represents the bit length of the second communication, and $|\sigmaOS|$ represents the bit length of the $1$-out-of-$n$ oblivious signature scheme.
In columns, $\lambda$ denotes a security parameter.
$|c^{\COM}|$ (resp. $|r^{\COM}|$) denotes the bit length of a commitment (resp. randomness) and $|\sigmaDS|$ (resp. $|\vkDS|$) denotes the bit length of a digital signature (resp. verification key) used to instantiate the $1$-out-of-$n$ oblivious signature scheme.
\label{CompareGenericOS}
\end{figure}

Particular, if the user interacts with the signer with a message list $M=(m_{i})_{i \in \{1,\dots, n\}}$ and the first communication message $\mu$, then the signer sends $n$ digital signatures $(\sigma^{\DS}_{i})_{i \in \{1,\dots, n\}}$ to the user as the second communication message where $\sigma^{\DS}_{i}$ is a signature on a message $(m_{i}, \mu)$.
This means that the second communication message cost (size) is proportional to $n$.

We improve the second communication cost by using a Merkle tree.
Concretely, instead of signing each $(m_{i}, \mu)$ where $m_{i} \in M$, we modify it to sign a message $(\sfroot,  \mu)$ where $\sfroot$ is a root of the Merkle tree computed from $M$.
By this modification, we reduce the communication cost of the second round from $n$ digital signatures to only one digital signature.
As a side effect of our modification, the size of the obtained $1$-out-of-$n$ oblivious signature is increasing, but it is proportional to $\log {n}$.
Our modification has the merit that the sum of a second communication message size and a signature size is improved from $O(n)$ to $O(\log{n})$.

\subsection{Road Map}
In Section \ref{SecPrelimi}, we introduce notations and review commitments, digital signatures, and Merkle tree.
In Section \ref{SecONOSRev}, we review $1$-out-of-$n$ oblivious signatures, revisit the definition of unforgeability by Tuo et al. \cite{TOO08}, and redefine unforgeability.
In Section \ref{SecOurConfromZLH}, we give a generic construction of $1$-out-of-$n$ oblivious signature schemes with efficient communication cost by improving the construction by Zhou et al. \cite{ZLH22} and prove security for our scheme.
In Section \ref{SecConclude}, we conclude our result and discuss open problems.

\section{Preliminaries}\label{SecPrelimi}
In this section, we introduce notations and review fundamental cryptographic primitives for constructing our $1$-out-of-$n$ oblivious signature scheme.

\subsection{Notations}
Let $1^{\lambda}$ be the security parameter. 
A function $f$ is negligible in $k$ if $f(k) \leq 2^{-\omega(\log k)}$.
For a positive integer $n$, we define $[n]: =\{1,\dots, n\}$.
For a finite set $S$, $s \xleftarrow{\$} S$ represents that an element $s$ is chosen from $S$ uniformly at random.

For an algorithm $\A$, $y \leftarrow \A(x)$ denotes that the algorithm $\A$ outputs $y$ on input~$x$.
When we explicitly show that $\A$ uses randomness $r$, we denote $y \leftarrow \A(x; r)$.
We abbreviate probabilistic polynomial time as PPT.

We use a code-based security game \cite{BR06}. 
The game $\sfGame$ is a probabilistic experiment in which adversary $\A$ interacts with an implied challenger $\C$ that answers oracle queries issued by $\A$. 
The $\sfGame$ has an arbitrary amount of additional oracle procedures which describe how these oracle queries are answered.
When the game $\sfGame$ between the challenger $\C$ and the adversary $\A$ outputs $b$, we write $\sfGame_{\A} \Rightarrow b$.
We say that $\A$ wins the game $\sfGame$ if $\sfGame_{\A} \Rightarrow 1$.
We implicitly assume that the randomness in the probability term $\Pr[\sfGame_{\A} \Rightarrow 1]$ is over all the random coins in the game.

\subsection{Commitment Scheme}
We review a commitment scheme and its security notion.
\begin{definition}[Commitment Scheme]
A commitment scheme $\COM$ consists of a following tuple of algorithms $(\nCOMKeyGen, \nCOMCommit)$.
\begin{itemize}
\item $\nCOMKeyGen (1^{\lambda}):$ A key-generation algorithm takes as an input a security parameter $1^{\lambda}$. It returns a commitment key $\ck$.
In this work, we assume that $\ck$ defines a message space, randomness space, and commitment space. We represent these space by  $\mathcal{M}_{\ck}$, $\Omega_{\ck}$, and $\mathcal{C}_{\ck}$, respectively.
\item $\nCOMCommit (\ck, m; r):$ A commit algorithm takes as an input a commitment key $\ck$, a message $m$, and a randomness $r$. It returns a commitment $c$.
In this work, we use the randomness $r$ as the decommitment (i.e., opening) information for $c$.
\end{itemize}
\end{definition}

\begin{definition}[Computational Hiding]
Let $\COM = (\nCOMKeyGen, \nCOMCommit)$ a commitment scheme and $\A$ a PPT algorithm.
We say that the $\COM$ satisfies computational hiding if for all $\A$, the following advantage of the hiding game
\begin{equation*}
\begin{split}
\Adv^{\Hide}_{\COM, \A}({\lambda}):= \left|
\Pr \left[
b = b^* \middle|
\begin{split}
&\ck \leftarrow \COMKeyGen(1^{\lambda}), (m_{0}, m_{1}, \st) \leftarrow \A(\ck),\\
&b \xleftarrow{\$} \{0, 1\},
c^* \leftarrow \COMCommit(\ck, m_{b}), b^* \leftarrow \A(c^*, \st)  
\end{split}
\right]
- \frac{1}{2}
\right|
\end{split}
\end{equation*}
is negligible in $\lambda$.

\end{definition}
\begin{definition}[Strong Computational Binding]
Let $\COM = (\nCOMKeyGen, \nCOMCommit)$ a commitment scheme and $\A$ a PPT algorithm.
We say that the $\COM$ satisfies strong computational binding if the following advantage 
\begin{equation*}
\begin{split}
&\Adv^{\sBind}_{\COM, \A}({\lambda}): = \Pr \left[
\begin{split}
&\nCOMCommit(\ck, m; r) = \nCOMCommit(\ck, m'; r')\\
 &\land (m, r) \neq (m', r') 
\end{split}
\middle|
\begin{split}
&\ck \leftarrow \nCOMKeyGen(1^{\lambda}),\\
&((m, r), (m', r')) \leftarrow\A(\ck)
\end{split}
 \right]
\end{split}
\end{equation*}
is negligible in $\lambda$.
\end{definition}

A commitment scheme with computational hiding and strong computational binding property can be constructed from a public key encryption (PKE) scheme with indistinguishable under chosen plaintext attack (IND-CPA) security.
We refer the reader to \cite{ZLH22} for a commitment scheme construction from a PKE scheme.

\subsection{Digital Signature Scheme}
We review a digital signature scheme and its security notion.
\begin{definition}[Digital Signature Scheme]
A digital signature scheme $\DS$ consists of following four algorithms $(\nDSSetup, \nDSKeyGen, \allowbreak \nDSSign, \nDSVerify)$.
\begin{itemize}
\item $\nDSSetup (1^{\lambda}):$ A setup algorithm takes as an input a security parameter $1^{\lambda}$. It returns the public parameter $\pp$.
In this work, we assume that $\pp$ defines a message space and represents this space by  $\mathcal{M}_{\pp}$.
We omit a public parameter $\pp$ in the input of all algorithms except for $\nDSKeyGen$.
\item $\nDSKeyGen (\pp):$ A key-generation algorithm takes as an input a public parameter $\pp$. It returns a verification key $\vk$ and a signing key $\sk$.
\item $\nDSSign (\sk, m):$ A signing algorithm takes as an input a signing key $\sk$ and a message $m$. It returns a signature $\sigma$.
\item $\nDSVerify (\vk, m, \sigma):$ A verification algorithm takes as an input a verification key $\vk$, a message $m$, and a signature $\sigma$.
It returns a bit $b \in  \{0, 1\}$.
\end{itemize}
\paragraph{\bf Correctness.}
$\DS$ satisfies correctness if for all $\lambda \in \N$, $\pp \leftarrow \nDSSetup (1^{\lambda})$ for all $m \in \mathcal{M}_{\pp}$, $(\vk, \sk) \leftarrow \nDSKeyGen(\pp)$, and $\sigma \leftarrow \nDSSign(\sk, m)$, $\nDSVerify(\vk, m, \sigma) = 1$ holds.
\end{definition}

We review a security notion called the strong existentially unforgeable under chosen message attacks $(\rmsEUFCMA)$ security for digital signature.

\begin{definition}[sEUF-CMA Security]
Let $\DS=(\nDSSetup, \nDSKeyGen, \allowbreak \nDSSign, \nDSVerify)$ be a signature scheme and $\A$ a PPT algorithm.
The strong existentially unforgeability under chosen message attacks (sEUF-CMA) security for $\DS$ is defined by the sEUF-CMA security game $\sfGame^{\sEUFCMA}_{\DS, \A}$ between the challenger $\C$ and $\A$ in Fig. \ref{sEUFCMAgame}.

\begin{figure}[h]
\centering
\begin{tabular}{|l|}
\hline
GAME $\sfGame^{\sEUFCMA}_{\DS, \A}(1^{\lambda}):$\\
~~~$\mathbb{L}^{\Sign} \leftarrow \{\}$, $\pp \leftarrow \nDSSetup (1^{\lambda})$, $(\vk, \sk) \leftarrow \nDSKeyGen(\pp)$, $(m^*, \sigma^*) \leftarrow \A^{\mathcal{O}^{\Sign}(\cdot)}(\pp, \vk)$\\
~~~If $\nDSVerify(\vk, m^*, \sigma^*) = 1 \land ~ (m^*, \sigma^*) \notin \mathbb{L}^{\Sign}$, return $1$. Otherwise return $0$.\\
\\
Oracle $\mathcal{O}^{\Sign}(m):$\\
~~~$\sigma \leftarrow \nDSSign(\sk, m)$, $\mathbb{L}^{\Sign} \leftarrow \mathbb{L}^{\Sign} \cup \{(m, \sigma)\}$, return $\sigma$.\\
\hline
\end{tabular}
\caption{\small The $\rmsEUFCMA$ security game $\sfGame^{\sEUFCMA}_{\DS, \A}$.}
\label{sEUFCMAgame}
\end{figure}

The advantage of an adversary $\A$ for the $\rmsEUFCMA$ security game is defined by $\Adv^{\sEUFCMA}_{\DS, \A}(\lambda)\allowbreak:= \Pr[\sfGame^{\sEUFCMA}_{\DS, \A}(1^{\lambda}) \Rightarrow 1]$.
$\DS$ satisfies $\rmsEUFCMA$ security if for all PPT adversaries $\A$, $\Adv^{\sEUFCMA}_{\DS, \A}({\lambda})$ is negligible in $\lambda$.
\end{definition}

\subsection{Merkle Tree Technique}\label{SubsecCRHandMer}
We review the collision resistance hash function family and the Merkle tree technique.

\begin{definition}[Collision Resistance Hash Function Family]
Let $\mathcal{H} = \{H_{\lambda}\}$ be a family of hash functions where $H_{\lambda} = \{H_{\lambda, i}: \{0, 1\}^* \rightarrow \{0, 1\}^{\lambda} \}_{i \in \mathcal{I}_{\lambda}}$.
$\mathcal{H}$ is a family of collision-resistant hash functions if for all PPT adversaries $\A$, the following advantage
\begin{equation*}
\Adv^{\Coll}_{\mathcal{H}, \A}({\lambda}):=\Pr [H(x) = H(x')|H \xleftarrow{\$} H_{\lambda}, (x, x') \leftarrow \A(H)]
\end{equation*}
is negligible in $\lambda$.
\end{definition}

\begin{definition}[Merkle Tree Technique  \cite{Mer87}]
The Merkle tree technique $\MT$ consists of following three algorithms $(\MerkleTree, \allowbreak \MerklePath,  \allowbreak \RootReconstruct)$ with access to a common hash function $H:\{0, 1\}^{*} \allowbreak \rightarrow \{0, 1\}^{\lambda}$.
\begin{itemize}
\item $\MerkleTree^{H} (M=(m_{0}, \dots, m_{2^{k}-1})):$ A Merkle tree generation algorithm takes as an input a list of $2^{k}$ elements $M=(m_{0}, \dots, m_{2^{k}-1})$.
It constructs a complete binary tree whose height is $k+1$ (i.e., maximum level is $k$).

We represent a root node as $w_{\epsilon}$ and a node in level $\ell$ as $w_{b_{1}, \dots b_{\ell}}$ where $b_{j} \in \{0, 1\}$ for $j \in [\ell]$.
The leaf node with an index $i \in \{0, \dots, 2^{k}-1\}$ represents $w_{\ItoB(i)}$ where $\ItoB$ is a conversion function from an integer $i$ to the $k$-bit binary representation. 

Each leaf node with an index $i \in \{0, \dots, 2^{k}-1\}$ (i.e., $w_{\ItoB(i)}$) is assigned a value $h_{\ItoB(i)} = H(m_{i})$.
Each level $j$ internal (non-leaf) node $w_{b_{1}, \dots b_{j}}$ is assigned a hash value $h_{b_{1}, \dots b_{j}} = H(h_{b_{1}, \dots b_{j},0}||h_{b_{1}, \dots b_{j},1})$ where $h_{b_{1}, \dots b_{j},0}$ and $h_{b_{1}, \dots b_{j},1}$ are values assigned to the left-children node $w_{b_{1}, \dots b_{j}, 0}$ and the right-children node $w_{b_{1}, \dots b_{j}, 1}$, respectively.
The root node $w_{\epsilon}$ is assigned a hash value $h_{\epsilon} = H(h_{0}||h_{1})$ and denote this value as $\sfroot$.
This algorithm outputs a value $\sfroot$ and the description $\sftree$ which describes the entire tree.

\item $\MerklePath^{H}(\sftree, i):$ A Merkle path generation algorithm takes as an input a description of a tree $\sftree$ and a leaf node index $i \in \{0, \dots, 2^{k}-1\}$.
Then, this algorithm computes $(b_{1}, \dots b_{k}) = \ItoB(i)$ and outputs a list $\sfpath = (h_{\overline{b_{1}}}, h_{b_{1}, \overline{b_{2}}}, \dots, h_{b_{1}, \dots, \overline{b_{k}}})$ where $\overline{b_{j}} = 1 - b_{j}$ for $j \in [k]$.

\item $\RootReconstruct^{H}(\sfpath, m_{i}, i):$ A root reconstruction algorithm takes as an input a list $\sfpath = (h_{\overline{b_{1}}}, h_{b_{1}, \overline{b_{2}}}, \dots, h_{b_{1}, \dots, \overline{b_{k}}})$, an element $m_{i}$, and a leaf node index $i \in \{0, \dots, 2^{k}-1\}$.
This algorithm computes $(b_{1}, \dots b_{k}) = \ItoB(i)$ and assigns $h_{b_{1}, \dots b_{k}}$.
For $i = k-1$ to $1$, computes $h_{b_{1}, \dots b_{j}} = H(h_{b_{1}, \dots b_{j},0}||h_{b_{1}, \dots b_{j},1})$ and outputs $\sfroot = H(h_{0}||h_{1})$.
\end{itemize}
\end{definition}

\begin{lemma}[Collision Extractor for Merkle Tree]
There exists the following efficient collision extractor algorithms $\Ext_{1}$ and $\Ext_{2}$.
\begin{itemize}
\item $\Ext_{1}$ takes as an input a description of Merkle tree $\sftree$ whose root node is assigned value $\sfroot$ and $(m_{i}', \sfpath, i)$.
If $\sftree$ is constructed from a list $M = (m_{0}, \dots, m_{2^{k}-1})$, $m_{i} \neq m'_{i}$, and $\sfroot = \RootReconstruct^{H}(\sfpath, m_{i}, i)$ holds, it outputs a collision of the hash function $H$.

\item $\Ext_{2}$ takes as an input a tuple $(m, j, \sfpath, \sfpath')$.
If $\sfpath \neq \sfpath'$ and $\RootReconstruct^{H}(\sfpath, \allowbreak m, \allowbreak j) \allowbreak = \RootReconstruct^{H}(\sfpath', m, j)$ hold, it outputs a collision of the hash function $H$.
\end{itemize}
\end{lemma}

\section{Security of Oblivious Signatures Revisited}\label{SecONOSRev}
In this section, first, we review a definition of a $1$-out-of-$n$ signature scheme and security notion called ambiguity.
Next, we review the security definition of the unforgeability in \cite{TOO08} and discuss the problems of their security model.
Then, we redefine the unforgeability security for a $1$-out-of-$n$ signature scheme.
  
\subsection{(1, n)-Oblivious Signature Scheme}

We review a syntax of a $1$-out-of-$n$ oblivious signature scheme and the security definition of ambiguity.

\begin{definition}[Oblivious Signature Scheme]
a $1$-out-of-$n$ oblivious signature scheme $\onOS$ consists of following algorithms $(\nOSSetup, \nOSKeyGen, \allowbreak \nOSUi, \nOSSii, \nOSUDer,  \allowbreak \nOSVerify)$.
\begin{itemize}
\item $\nDSSetup (1^{\lambda}):$ A setup algorithm takes as an input a security parameter $1^{\lambda}$. It returns the public parameter $\pp$.
In this work, we assume that $\pp$ defines a message space and represents this space by  $\mathcal{M}_{\pp}$.
We omit a public parameter $\pp$ in the input of all algorithms except for $\nOSKeyGen$.
\item $\nOSKeyGen (\pp):$ A key-generation algorithm takes as an input a public parameter $\pp$. It returns a verification key $\vk$ and a signing key $\sk$.
\item $\nOSUi (\vk, M=(m_{0}, \dots, m_{n-1}), j):$ This is a first message generation algorithm that is run by a user. It takes as input a verification key $\vk$, a list of message $M=(m_{0}, \dots, m_{n-1})$, and a message index $j \in \{0, \dots, n-1\}$.
It returns a pair of a first message and a state $(\mu, \st)$ or $\bot$.
\item $\nOSSii (\vk, \sk, M=(m_{0}, \dots, m_{n-1}), \mu):$ This is a second message generation algorithm that is run by a signer. It takes as input a verification key $\vk$, a signing key $\sk$, a list of message  $M=(m_{0}, \dots, m_{n-1})$, and a first message~$\mu$.
It returns a second message $\rho$ or $\bot$.
\item $\nOSUDer(\vk, \st, \rho):$ This is a signature derivation algorithm that is run by a user. 
It takes as an input a verification key $\vk$, a state $\st$, and a second message $\rho$.
It returns a pair of a message and its signature $(m, \sigma)$ or $\bot$. 
\item $\nOSVerify (\vk, m, \sigma):$ A verification algorithm takes as an input a verification key $\vk$, a message $m$, and a signature $\sigma$.
It returns a bit $b \in  \{0, 1\}$.
\end{itemize}
\paragraph{\bf Correctness.}
$\onOS$ satisfies correctness if for all $\lambda \in \N$, $n \leftarrow n(\lambda)$, $\pp \leftarrow \nOSSetup (1^{\lambda})$, for all message set $ \mathcal{M}=(m_{0}, \dots, m_{n-1})$ such that $m_{i} \in \mathcal{M}_{\pp}$,  $(\vk, \sk) \leftarrow \nOSKeyGen(\pp)$, for all $j \in \{0, \dots n-1\}$, $(\mu, \st) \leftarrow \nOSUi(\vk, M, j)$, $\rho \leftarrow \nOSSii (\vk, \sk, M, \mu)$, and $(m_{j}, \sigma) \leftarrow \nOSUDer(\vk, \st, \rho)$, $\nDSVerify(\vk, m_{j}, \sigma) = 1$ holds.
\end{definition}

\begin{definition}[Ambiguity]
Let $\onOS = (\nOSSetup, \nOSKeyGen, \allowbreak \nOSUi, \nOSSii, \nOSUDer, \nOSVerify)$ be an oblivious signature scheme and $\A$ a PPT algorithm.
The ambiguity for $\onOS$ is defined by the ambiguity security game $\sfGame^{\Amb}_{\onOS, \A}$ between the challenger $\C$ and $\A$ in Fig. \ref{Ambgame}.

\begin{figure}[h]
\centering
\begin{tabular}{|l|}
\hline
GAME $\sfGame^{\Amb}_{\onOS, \A}(1^{\lambda}):$\\
~~~$\pp \leftarrow \nOSSetup (1^{\lambda})$, $(\vk, \sk) \leftarrow \nOSKeyGen(\pp)$,\\
~~~$(M=(m_{0}, \dots, m_{n-1}), i_{0}, i_{1}, \st_{\A}) \leftarrow \A(\pp, \vk, \sk)$\\
~~~$b \xleftarrow{\$} \{0, 1\}$, $(\mu, \st_{\sfS}) \leftarrow \nOSUi(\vk, M, i_{b})$, $b^* \leftarrow \A(\mu, \st_{\A})$.\\
~~~If $b^* = b$ return $1$. Otherwise return $0$.\\
\hline
\end{tabular}
\caption{\small The ambiguity security game $\sfGame^{\Amb}_{\onOS, \A}$.}
\label{Ambgame}
\end{figure}

The advantage of an adversary $\A$ for the ambiguity security game is defined by $\Adv^{\Amb}_{\onOS, \A}({\lambda}) \allowbreak:= | \Pr[\sfGame^{\Amb}_{\onOS, \A}(1^{\lambda}) \Rightarrow 1] - \frac{1}{2} |$.
$\onOS$ satisfies ambiguity if for all PPT adversaries $\A$, $\Adv^{\Amb}_{\onOS, \A}({\lambda})$ is negligible in $\lambda$.
\end{definition}

\subsection{Definition of Unforgeability Revisited}\label{Subsec_Def_Unf_Prev}

We review the security definition of unforgeability for $\onOS$ in previous works in \cite{TOO08}.
The unforgeability for a $1$-out-of-$n$ oblivious signature scheme in  \cite{TOO08} is formalized by the following game between a challenger $\C$ and a PPT adversary~$\A$.
\begin{itemize}
\item $\C$ runs $\pp \leftarrow \nOSSetup (1^{\lambda})$ and $(\vk, \sk) \leftarrow \nOSKeyGen(\pp)$, and gives $(\pp, \vk)$ to $\A$.
\item $\A$ is allowed to engage polynomially many  signing protocol executions.

In an $i$-th $(1 \leq i \leq t)$ protocol execution,  
\begin{itemize}
\item $\A$ makes a list $M_{i}=(m_{i, 0}, \dots, m_{i, n-1})$ and chooses $m_{i, j_{i}}$.
\item $\A$ sends $(\mu_{i}, M_{i}=(m_{i, 0}, \dots, m_{i, n-1}))$ to $\C$. 
\item $\C$ runs $\rho_{i} \leftarrow  \nOSSii (\vk, \sk, M_{i}, \mu_{i})$ and gives $\rho_{i}$ to $\A$. 
\end{itemize}
\item Let $\mathbb{L}^{\Sign}=\{m_{1, j_{1}}, \dots m_{t, j_{t}}\}$ be a list of messages that $\A$ has obtained signatures.
$\A$ outputs a forgery $(m^*, \sigma^*)$ which satisfies $m^* \notin \mathbb{L}^{\Sign}$. $\A$ must complete all singing executions before it outputs a forgery.
\end{itemize}
If no PPT adversary $\A$ outputs a valid forgery in negligible probability in $\lambda$, $\onOS$ satisfies the unforgeability security.

We point out three problems for the above security definition.
\begin{itemize}
\item{\bf Problem 1: How to Store Messages in $\mathbb{L}^{\Sign}$:}
In the above security game, $\C$ needs to store corresponding messages that $\A$ has obtained signatures.
However, by ambiguity property, $\C$ cannot identify the message $m_{i, j_{i}}$ which is selected by the $\A$ from a transcription of the $i$-th interaction with $M_{i}$.
This security model does not explain how to record an entry of $\mathbb{L}^{\Sign}$.

\item{\bf Problem 2:  Trivial Attack:}
Let us consider the following adversary $\A$ that runs signing protocol execution twice.
$\A$ chooses $M=(m_{0}, m_{1})$ where $m_{0}$ and $m_{1}$ are distinct, and sets lists as $M_{1}=M_{2} = M$.
In the 1st interaction, $\A$ chooses $m_{0} \in M_{1}$, obtains a signature $\sigma_{0}$ on a message $m_{0}$, and outputs $(m_{0}, \sigma_{0})$ at the end of interaction.
In the 2nd interaction, $\A$ chooses $m_{1} \in M_{2}$, obtains a signature $\sigma_{1}$ on a message $m_{1}$, and outputs $(m_{0}, \sigma_{0})$ at the end of interaction.
Then, $\A$ outputs a trivial forgery $(m^*, \sigma^*) = (m_{1}, \sigma_{1})$.
This attack is caused by the resubmitting of a signature $(m_{0}, \sigma_{0})$ at the end of the signing interaction.
$\A$ pretends to have obtained $(m_{0}, \sigma_{0})$ in the 2nd signing interaction.
However, there is no countermeasure in the security model.

\item{\bf Problem 3: Missing Adversary Strategy:}
The security game does not capture an adversary with the following strategy.
Let us consider an adversary $\A$ that executes the signing query only once.
$\A$ interacts with the signer with a message list $M$ and intends to forge a signature $\sigma^*$ on a message $m^* \notin M$, but give up outputting $(m, \sigma)$ where $m \in M$ at the end of signing query.
Since the security game only considers the adversary that outputting $(m, \sigma)$ where $m \in M$ at the end of the signing query, the security game cannot capture the this adversary $\A$.

\end{itemize}

\subsection{New Unforgeability Definition}\label{SubsecNewUnf}
We modify the unforgeability security model by Tso et al. \cite{TOO08} by addressing the problems in the previous sections in order and redefine the unforgeability security. 
Here, we briefly explain how to address these problems.

\begin{itemize}
\item{\bf Countermeasure for Problem 1:}
This problem is easy to fix by forcing $\A$ to output $(m_{i, j_{i}}, \sigma_{i})$ at the end of each signing interaction.

\item{\bf Countermeasure for Problem 2:}
This attack is caused by the reuse of a signature at the end of signing interactions.
That is $\A$ submits $(m, \sigma)$ twice or more at the end of signing interactions.

To address this problem, we introduce the signature resubmission check.
This prevents resubmission of $(m, \sigma)$ at the end of signing interactions. 
However, this is not enough to prevent the reuse of a signature.
For example, if an oblivious signature has a re-randomizable property (i.e., The property that a signature is refreshed without the signing key),  $\A$ can easily pass resubmission checks by randomizing a obtained signature $ \sigma$ to $ \sigma'$.

For this reason, normal unforgeability security is not enough.
We address this issue by letting strong unforgeability security be a default for the security requirement.
\item{\bf Countermeasure for Problem 3:}
This problem is addressed by adding another winning path for $\A$.
When $\A$ submits $(m^*, \sigma^*)$ at the end of $i$-th signing interaction, if $(m^*, \sigma^*)$ submitted by $\A$ at the end of signing qyery is valid and $m^* \notin M_{i}$, $\A$ wins the game where $M_{i}$ is a list of messages send by $\A$ at the beginning of the $i$-th signing query.
\end{itemize}

By reflecting the above countermeasures to the unforgeability security model by Tso et al. \cite{TOO08}, we redefine the unforgeability security model as the strong unforgeability under chosen message attacks under the sequential signing interaction $(\rmSeqsEUFCMA)$ security.

\begin{definition}[Seq-sEUF-CMA Security]\label{Def_Seq-sEUF-CMA_Security}
Let $\onOS = (\nOSSetup, \nOSKeyGen, \allowbreak \nOSUi, \allowbreak \nOSSii, \nOSUDer, \allowbreak \nOSVerify)$ be a $1$-out-of-$n$ oblivious signature scheme and $\A$ a PPT algorithm.
The strong unforgeability under chosen message attacks under the sequential signing interaction $(\rmSeqsEUFCMA)$ security for $\onOS$ is defined by the $\rmSeqsEUFCMA$ security game $\sfGame^{\SeqsEUFCMA}_{\onOS, \A}$ between the challenger $\C$ and $\A$ in Fig. \ref{EUFCMAgame}.

\begin{figure}[h]
\centering
\begin{tabular}{|l|}
\hline
GAME $\sfGame^{\SeqsEUFCMA}_{\onOS, \A}(1^{\lambda}):$\\
~~~$\mathbb{L}^{\Sign} \leftarrow \{\}$, $\mathbb{L}^{\ListM} \leftarrow \{\}$, $q^{\Sign} \leftarrow 0$, $q^{\Fin} \leftarrow 0$\\
~~~$\pp \leftarrow \nOSSetup (1^{\lambda})$, $(\vk, \sk) \leftarrow \nOSKeyGen(\pp)$, $(m^*, \sigma^*) \leftarrow \A^{\mathcal{O}^{\Sign}(\cdot, \cdot), \mathcal{O}^{\Fin}(\cdot, \cdot)}(\pp, \vk)$\\
~~~If $q^{\Sign}  = q^{\Fin} \land \nOSVerify(\vk, m^*, \sigma^*) = 1 \land  (m^*, \sigma^*) \notin \mathbb{L}^{\Sign} $, return $1$. \\
~~~Otherwise return $0$.\\
\\
Oracle $\mathcal{O}^{\Sign}(M_{q^{\Sign}}, \mu):$\\
~~~If $q^{\Sign} \neq q^{\Fin}$, return $\bot$.\\
~~~$\rho \leftarrow \nOSSii (\vk, \sk, M, \mu)$, if $\rho = \bot$, return $\bot$.\\ 
~~~If $\rho \neq \bot$, $q^{\Sign} \leftarrow q^{\Sign} +1$, $\mathbb{L}^{\ListM} \leftarrow \mathbb{L}^{\ListM} \cup \{(q^{\Sign}, M_{q^{\Sign}})\}$, return $\rho$.\\
Oracle $\mathcal{O}^{\Fin}(m^*, \sigma^*):$\\
~~~If $q^{\Sign} \neq q^{\Fin} + 1$, return $\bot$.\\
~~~If $\nOSVerify(\vk, m^*, \sigma^*) = 0$, return the game output $0$ and abort.\\
~~~(//Oblivious signature resubmission check)\\
~~~\fbox{If $(m^*, \sigma^*) \in \mathbb{L}^{\Sign}$, return the game output $0$ and abort.}\\ 
~~~Retrieve an entry $(q^{\Sign}, M_{q^{\Sign}}) \in \mathbb{L}^{\ListM}$.\\
~~~If $m^* \in M_{q^{\Sign}}$, $\mathbb{L}^{\Sign} \leftarrow \mathbb{L}^{\Sign}  \cup  \{(m^*, \sigma^*)\}$, $q^{\Fin} \leftarrow q^{\Fin} + 1$, return ``$\mathtt{accept}$".\\
~~~(//Capture adversaries that give up completing signing executions in the game.)\\
~~~\fbox{If $m^* \notin M_{q^{\Sign}}$, return  the game output $1$ and abort. }
\begin{tabular}{l}
\vspace{10pt}
\end{tabular}\\
\hline
\end{tabular}
\caption{\small The $\rmSeqsEUFCMA$ security game $\sfGame^{\SeqsEUFCMA}_{\onOS, \A}$. The main modifications from previous works security game are highlighted in \fbox{white box}.}
\label{EUFCMAgame}
\end{figure}

The advantage of $\A$ for the $\rmSeqsEUFCMA$ security game is defined by $\Adv^{\SeqsEUFCMA}_{\onOS, \A}({\lambda}):= \Pr[\sfGame^{\SeqsEUFCMA}_{\onOS, \A}(1^{\lambda}) \Rightarrow 1]$.
$\onOS$ satisfies $\rmSeqsEUFCMA$ security if for all PPT adversaries $\A$, $\Adv^{\SeqsEUFCMA}_{\onOS, \A}({\lambda})$ is negligible in $\lambda$.
\end{definition}

Our security model is the sequential signing interaction model.
One may think that it is natural to consider the concurrent signing interaction model.
However, by extending our model to the concurrent signing setting there is a trivial attack.
We discuss the security model that allows concurrent signing interaction in Section~\ref{SecConclude}.

\section{Our Construction}\label{SecOurConfromZLH}
In this section, first, we review the generic construction by Zhou et al. \cite{ZLH22}.
Second, we propose our new generic construction based on their construction.
Then, we prove the security of our proposed scheme.

\subsection{Generic Construction by Zhou et al. \cite{ZLH22}}
The generic construction of a $1$-out-of-$n$ signature scheme $\onOS_{\ZLH}$ by Zhou et al. \cite{ZLH22} is a combination of a commitment scheme $\COM$ and a digital signature scheme $\DS$.
Their construction $\onOS_{\ZLH}[\COM, \DS] = (\OSSetup, \OSKeyGen, \allowbreak \OSUi, \OSSii, \OSUDer, \allowbreak \OSVerify)$ is given in Fig. \ref{ZLH_OSConst}.
\begin{figure}[h]
\centering
\begin{tabular}{|l|}
\hline
$\OSSetup(1^\lambda):$\\
~~~$\ck \leftarrow \COMKeyGen (1^{\lambda})$, $\ppDS \leftarrow \DSSetup (1^{\lambda})$, return $\ppOS \leftarrow (\ck, \ppDS)$.\\
$\OSKeyGen(\ppOS=(\ck, \ppDS)):$\\
~~~$(\vkDS, \skDS) \leftarrow \DSKeyGen(\ppDS)$, return $(\vkOS, \skOS) \leftarrow (\vkDS, \skDS)$.\\
$\OSUi(\vkOS, M=(m_{0}, \dots, m_{n-1}), j \in \{0, \dots, n-1\}):$\\
~~~$r \xleftarrow{\$} \Omega_{\ck}$, $c \leftarrow \COMCommit(\ck, m; r)$, $\mu \leftarrow c$, $\st \leftarrow (M, c, r, j)$, return $(\mu, \st)$.\\
$\OSSii(\vkOS, \skOS=\skDS, M=(m_{0}, \dots, m_{n-1}), \mu=c):$\\
~~~For $i = 0$ to $n-1$, $\sigmaDS_{i} \leftarrow \DSSign(\skDS, (m_{i}, c))$.\\
~~~Return $\rho \leftarrow (\sigmaDS_{0},\dots, \sigmaDS_{n-1})$.\\
$\OSUDer(\vkOS=\vkDS, \st=(M=(m_{0}, \dots, m_{n-1}), c, r, j), \rho = (\sigmaDS_{0},\dots,\sigmaDS_{n-1})):$\\
~~~For $i = 0$ to $n-1$, if $\DSVerify(\vkDS, (m_{i},c), \sigmaDS_{i})=0$, return $\bot$.\\
~~~$\sigmaOS \leftarrow (c, r, \sigmaDS_{j})$, return $(m_{j}, \sigmaOS)$.\\
$\OSVerify (\vkOS=\vkDS, m, \sigmaOS=(c, r, \sigmaDS):$\\
~~~If $c \neq \COMCommit(\ck, m; r)$, return $0$.\\
~~~If $\DSVerify(\vkDS, (m,c), \sigmaDS)=0$, return $0$.\\
~~~Otherwise return $1$.\\
\hline
\end{tabular}
\caption{\small The generic construction $\onOS_{\ZLH}[\COM, \DS]$. }
\label{ZLH_OSConst}
\end{figure}

We briefly provide an overview of a signing interaction and an intuition for the security of their construction.
In the signing interaction, the user chooses a message list $M= (m_{i})_{i \in \{0, \dots, n-1\}}$ and a specific message $m_{j_{i}}$ that the user wants to obtain the corresponding signature.
To hide this choice from the signer, the signer computes the commitment $c$ on $m_{j}$ with the randomness $r$.
The user sends $(M, \mu=c)$ to the signer.

Here, we provide an intuition for the security of their construction.
From the view of the signer, by the hiding property of the commitment scheme, the signer does not identify $m_{j}$ from $(M, \mu=c)$.
This guarantees the ambiguity of their construction.
The signer computes a signature $\sigmaDS_{i}$ on a tuple $(m_{i}, c)$ for $i\in \{0, \dots, n-1\}$ and sends $\rho = ( \sigmaDS_{i})_{i \in \{0, \dots, n-1\}}$.

If the signer honestly computes $c$ on $m_{j} \in M$, we can verify that $m_{j_{i}}$ is committed into $c$ by decommitting with $r$.
An oblivious signature on $m_{j}$ is obtained as $\sigmaOS = (c, r, \sigmaDS_{j})$.
If a malicious user wants to obtain two signatures for two distinct messages $m, m' \in M$ or obtain a signature on  $m^* \notin M$ from the signing protocol execution output $(M=(m_{i})_{i \in \{0, \dots, n-1\}}, \mu, \rho)$, the malicious user must break either the sEUF-CMA security of $\DS$ or the strong binding property of $\COM$.
This guarantees the unforgeability security of their construction.

A drawback of their construction is the second communication cost.
A second message $\rho$ consists of $n$ digital signatures.
If $n$ becomes large, it will cause heavy communication traffic.
It is desirable to reduce the number of signatures in $\rho$.

\subsection{Our Generic Construction}\label{SubsecOurGenCon}
As explained in the previous section, the drawback of the construction by Zhou et al. \cite{ZLH22} is the size of a second message $\rho$.
To circumvent this bottleneck, we improve their scheme by using the Merkle tree technique.
Concretely, instead of signing on $(m_{i}, c)$ for each $m_{i} \in M$, we modify it to sign on $(\sfroot,  c)$ where $\sfroot$ is a root of the Merkle tree computed from $M$.
This modification allows us to reduce the number of digital signatures included in $\rho$ from $n$ to~$1$.

Now, we describe our construction.
Let $\COM$ be a commitment scheme, $\DS$ a digital signature scheme, $\mathcal{H}=\{H_{\lambda}\}$ a hash function family, and $\MT = (\MerkleTree, \allowbreak \MerklePath,  \allowbreak \RootReconstruct)$ a Merkle tree technique in Section \ref{SubsecCRHandMer}.
To simplify the discussion, we assume that $n>1$ is a power of $2$. \footnote{With the following modification, our scheme also supports the case where $n>1$ is not a power of 2. Let $k$ be an integer such that $2^{k-1}< n < 2^{k}$. We change a list of message $M = (m_{0}, \dots m_{n-1})$ which is given to $\OSUi$ and $\OSSii$ as a part of an input to an augmented message list $M' = (m'_{0}, \dots m'_{2^{k}-1})$ where $m'_{i} = m_{i}$ for $i \in \{0,\dots, n-1\}$, $m'_{n-1+i} = \phi || i$ for $i \in \{1, \dots 2^k -n \}$, and $\phi$ is a special symbol representing that a message is empty.}

Our generic construction $\onOS_{\Ours}[\mathcal{H}, \COM, \DS] = (\OSSetup, \OSKeyGen, \allowbreak \OSUi,\allowbreak \OSSii, \allowbreak \OSUDer, \allowbreak \OSVerify)$ is given in Fig. \ref{Our_OSConst}.
\begin{figure}[h]
\centering
\begin{tabular}{|l|}
\hline
$\OSSetup(1^\lambda):$\\
~~~$H \xleftarrow{\$} H_{\lambda}$, $\ck \leftarrow \COMKeyGen (1^{\lambda})$, $\ppDS \leftarrow \DSSetup (1^{\lambda})$,\\
~~~Return $\ppOS \leftarrow (H, \ck, \ppDS)$.\\
$\OSKeyGen(\ppOS=(\ck, \ppDS)):$\\
~~~$(\vkDS, \skDS) \leftarrow \DSKeyGen(\ppDS)$, return $(\vkOS, \skOS) \leftarrow (\vkDS, \skDS)$.\\
$\OSUi(\vkOS, M=(m_{0}, \dots, m_{n-1}), j \in \{0, \dots, n-1\}):$\\
~~~If there exists $(t, t') \in \{0, \dots, n-1\}^2$ s.t. $t \neq t' \land m_{t} = m_{t'}$, return $\bot$.\\
~~~$r \xleftarrow{\$} \Omega_{\ck}$, $c \leftarrow \COMCommit(\ck, m; r)$, $\mu \leftarrow c$, $\st \leftarrow (M, c, r, j)$, return $(\mu, \st)$.\\
$\OSSii(\vkOS, \skOS=\skDS, M=(m_{0}, \dots, m_{n-1}), \mu=c):$\\
~~~If there exists $(t, t') \in \{0, \dots, n-1\}^2$ s.t.  $t \neq t' \land m_{t} = m_{t'}$, return $\bot$.\\
~~~$(\sfroot, \sftree) \leftarrow \MerkleTree^{H}(M)$, $\sigmaDS \leftarrow \DSSign(\skDS, (\sfroot, c))$.\\
~~~Return $\rho \leftarrow \sigmaDS$.\\
$\OSUDer(\vkOS=\vkDS, \st=(M=(m_{0}, \dots, m_{n-1}), c, r, j), \rho = (\sigmaDS_{1},\dots,\sigmaDS_{n})):$\\
~~~$(\sfroot, \sftree) \leftarrow \MerkleTree^{H}(M)$, $\sfpath \leftarrow \MerklePath^{H}(\sftree, j)$ \\
~~~If $\DSVerify(\vkDS, (\sfroot,c), \sigmaDS)=0$, return $\bot$.\\
~~~$\sigmaOS \leftarrow (\sfroot, c, \sigmaDS, \sfpath, j, r)$, return $(m_{j}, \sigmaOS)$.\\
$\OSVerify (\vkOS=\vkDS, m, \sigmaOS=(\sfroot, c, \sigmaDS, \sfpath, j, r):$\\
~~~If $\sfroot \neq \RootReconstruct^{H}(\sfpath, m, j)$, return $0$.\\
~~~If $c \neq \COMCommit(\ck, m; r)$, return $0$.\\
~~~If $\DSVerify(\vkDS, (\sfroot,c), \sigmaDS)=0$, return $0$.\\
~~~Otherwise return $1$.\\
\hline
\end{tabular}
\caption{\small Our generic construction $\onOS_{\Ours}[\mathcal{H}, \COM, \DS]$. }
\label{Our_OSConst}
\end{figure}

\subsection{Analysis}
We analyze our scheme $\onOS_{\Ours}$.
It is easy to see that our scheme satisfies the correctness.
Now, we prove that our generic construction $\onOS_{\Ours}$ satisfies the ambiguity and the $\rmSeqsEUFCMA$ security.

\begin{theorem}\label{OurconAmb}
If $\COM$ is computational hiding commitment, $\onOS_{\Ours}[\mathcal{H}, \allowbreak \COM, \allowbreak \DS]$ satisfies the ambiguity.
\end{theorem}
\begin{proof}
The ambiguity of our scheme can be proven in a similar way in \cite{ZLH22}.
Let $\A$ be an adversary for the ambiguity game of $\onOS_{\Ours}$.
We give a reduction algorithm $\B$ that reduces the ambiguity security of our scheme to the computational hiding property of $\COM$ in Fig. \ref{AmbCOMSimOSour}. 

\begin{figure}[htbp]
\centering
\begin{tabular}{|l|}
\hline
$\B(1^{\lambda}, \ck):$\\
~~~$H \xleftarrow{\$} H_{\lambda}$, $\ppDS \leftarrow \DSSetup (1^{\lambda})$, $\ppOS \leftarrow (H, \ck, \ppDS)$, \\
~~~$(\vkDS, \skDS) \leftarrow \DSKeyGen(\ppDS)$, $(\vkOS, \skOS) \leftarrow (\vkDS, \skDS)$,  \\
~~~$(M=(m_{0}, \dots, m_{n-1}), i_{0}, i_{1}, \st_{\A}) \leftarrow \A(\ppOS, \vkOS, \skOS)$\\
~~~$m^*_{0} \leftarrow m_{i_{0}}$, $m^*_{1} \leftarrow  m_{i_{1}}$, send $(m^*_{0}, m^*_{1})$ to the challenger $\C$ and obtain  \\
~~~~~~a commitment $c^*$ where $c^* \leftarrow \COMCommit(\ck, m^*_{b})$ and $b \xleftarrow{\$} \{0, 1\}$ is chosen $\C$.\\
~~~$b' \leftarrow \A(\mu=c^*, \st_{\A})$, return $b^* \leftarrow b'$.\\
\hline
\end{tabular}
\caption{\small The reduction algorithm $\B$.}
\label{AmbCOMSimOSour}
\end{figure}

Now, we confirm that $\B$ simulates the ambiguity game of $\onOS_{\Ours}$.
In the case that $b=0$, $c^*\leftarrow \COMCommit(\ck, m^*_{0} = m_{i_{0}})$ holds.
$\B$ simulates $\mu$ on the choice of $m_{i_{0}}$ in this case.
Similarly, in the case that $b=1$, $c^*\leftarrow \COMCommit(\ck, m^*_{1} = m_{i_{1}})$ holds.
$\B$ simulates $\mu$ on the choice of $m_{i_{1}}$ in this case.
Since $b$ is chosen uniformly at random from $\{0, 1\}$, $\B$ perfectly simulates the ambiguity game of $\onOS_{\Ours}$.
We can see that $\Adv^{\Hide}_{\COM, \B} ({\lambda})= \Adv^{\Amb}_{\onOS_{\Ours}, \A}({\lambda})$ holds.
Thus, we can conclude Theorem \ref{OurconAmb}.
\qed
\end{proof}

\begin{theorem}\label{OurconEUF}
If $\mathcal{H}$ is a family of collision-resistant hash functions, $\DS$ satisfies the $\rmsEUFCMA$ security, and $\COM$ is a strong computational binding commitment, $\onOS_{\Ours}[\mathcal{H}, \COM, \DS]$ satisfies the $\rmSeqsEUFCMA$ security.
\end{theorem}

\begin{proof}
Let $\A$ be a PPT adversary for the $\rmSeqsEUFCMA$ game of $\onOS_{\Ours}$.
We introduce the base game $\sfGame^{\sfBasic}_{\onOS_{\Ours}, \A}$ which simulates $\sfGame^{\SeqsEUFCMA}_{\onOS_{\Ours}, \A}$.
We provide $\sfGame^{\sfBasic}_{\onOS_{\Ours}, \A}$ in Fig.~\ref{BasicSimOSour}.

\begin{figure}[htbp]
\centering
\begin{tabular}{|l|}
\hline
$\sfGame^{\sfBasic}_{\onOS_{\Ours}, \A}(1^{\lambda}):$ \\
~~~$\mathbb{L}^{\Sign} \leftarrow \{\}$, $\mathbb{L}^{\ListM} \leftarrow \{\}$, $\mathbb{T} \leftarrow \{\}$, $q^{\Sign} \leftarrow 0$, $q^{\Fin} \leftarrow 0$,  $\Final \leftarrow \ttfalse$,   \\
~~~$\ReuseDS \leftarrow \ttfalse$, $\CollCOM \leftarrow \ttfalse$, $\ForgeDS\leftarrow \ttfalse$, $\ck \leftarrow \COMKeyGen (1^{\lambda})$,  \\
~~~$\ppDS \leftarrow \DSSetup (1^{\lambda})$, $\ppOS \leftarrow (H, \ck, \ppDS)$,  $(\vkDS, \skDS) \leftarrow \DSKeyGen(\ppDS)$, \\
~~~$(\vkOS, \skOS) \leftarrow (\vkDS, \skDS)$, $(m^*, \sigmastrOS) \leftarrow \A^{\mathcal{O}^{\Sign}(\cdot, \cdot), \mathcal{O}^{\Fin}(\cdot, \cdot)}(\ppOS, \vkOS)$\\
~~~If $q^{\Sign} \neq q^{\Fin} \lor \OSVerify(\vkOS, m^*, \sigma^*) \neq 1 \lor (m^*, \sigmastrOS) \in \mathbb{L}^{\Sign}$, return $0$.\\
~~~$\Final \leftarrow \tttrue$, $\mathbb{L}^{\Sign} \leftarrow \mathbb{L}^{\Sign}  \cup  \{(m^*, \sigmastrOS)\}$, $q^{\Fin} \leftarrow q^{\Fin} + 1$\\
~~~Search a pair $(\widetilde{m}, \widetilde{\sigma}^{\OS}) \neq (\widetilde{m}', \widetilde{\sigma}'{}^{\OS})$ in $\mathbb{L}^{\Sign}$ such that the first three   \\  
~~~~~~elements of $\sigmaOS$ are the same.  i.e.,  $(\widetilde{\sfroot}, \widetilde{c}, \widetilde{\sigma}^{\DS}) = (\widetilde{\sfroot}', \widetilde{c}', \widetilde{\sigma}'{}^{\DS})$ \\
~~~If there is no such a pair, $\ForgeDS\leftarrow \tttrue$, return $1$.\\
~~~~~~~~~~~$(\Final = \tttrue \land \ReuseDS = \ttfalse \land \ForgeDS= \tttrue \land  \CollCOM = \ttfalse)$\\
~~~$\ReuseDS \leftarrow \tttrue$.\\
~~~Parse $\widetilde{\sigma}^{\OS}$ as $(\sfroot^*, c^*, \sigmastrDS, \widetilde{\sfpath}, \widetilde{j}, \widetilde{r})$,
$\widetilde{\sigma}'{}^{\OS}$ as $(\sfroot^*, c^*, \sigmastrDS, \widetilde{\sfpath}', \widetilde{j}', \widetilde{r}')$.\\
~~~If $(\widetilde{m}, \widetilde{r})  \neq (\widetilde{m}', \widetilde{r}')$, $\CollCOM \leftarrow \tttrue$, return $1$. \\
~~~~~~~~~~~$(\Final = \tttrue \land \ReuseDS = \tttrue \land  \ForgeDS= \ttfalse  \land \CollCOM = \tttrue)$\\
~~~Otherwise, return $1$. \\
~~~~~~~~~~~$(\Final = \tttrue \land \ReuseDS = \tttrue \land  \ForgeDS= \ttfalse  \land\CollCOM = \ttfalse)$\\
\\
Oracle $\mathcal{O}^{\Sign}(M = (m_{0}, \dots, m_{n-1}), \mu=c):$\\
~~~If $q^{\Sign} \neq q^{\Fin}$, return $\bot$.\\
~~~If there exists a pair $(t \neq t' \in \{0, \dots, n-1\})$ such that $m_{t} = m_{t'}$, return $\bot$.\\
~~~$(\sfroot, \sftree) \leftarrow \MerkleTree^{H}(M)$, $\sigmaDS \leftarrow \DSSign(\skDS, (\sfroot, c))$,\\
~~~$q^{\Sign} \leftarrow q^{\Sign}+1$, $M_{q^{\Sign}} \leftarrow  M$, $\mathbb{L}^{\ListM} \leftarrow \mathbb{L}^{\ListM} \cup \{(q^{\Sign}, M_{q^{\Sign}})\}$, \\
~~~$\mathbb{T} \leftarrow \mathbb{T} \cup \{(q^{\Sign}, M_{q^{\Sign}}, \sfroot, c,\sigmaDS) \}$,\\
~~~return $\rho \leftarrow \sigmaDS$ to $\A$. \\
Oracle $\mathcal{O}^{\Fin}(m^*, \sigmastrOS):$\\
~~~If $q^{\Sign} \neq q^{\Fin} + 1$, return $\bot$.\\
~~~If $\OSVerify(\vkOS, m^*, \sigmastrOS) \neq 1$, return the game output $0$ and abort.\\
~~~If $(m^*, \sigmastrOS) \in \mathbb{L}^{\Sign}$, return the game output $0$ and abort.\\ 
~~~$\mathbb{L}^{\Sign} \leftarrow \mathbb{L}^{\Sign}  \cup  \{(m^*, \sigmastrOS)\}$, $q^{\Fin} \leftarrow q^{\Fin} + 1$,  retrieve $(q^{\Sign}, M_{q^{\Sign}}) \in \mathbb{L}^{\ListM}$.\\
~~~If $m^* \in M_{q^{\Sign}}$, return ``$\mathtt{accept}$"  to $\A$. \\
~~~Parse $\sigmastrOS$ as $(\sfroot^*, c^*, \sigmastrDS, \sfpath^*, j^*,  r^*)$.\\
~~~If $(q^{\Sign}, *,  \sfroot^*, c^*, \sigmastrDS) \in \mathbb{T}$ return the game output $1$.\\
~~~~~~~~~~~$(\Final = \ttfalse \land \ReuseDS = \ttfalse \land \ForgeDS = \ttfalse \land \CollCOM = \ttfalse)$\\
~~~Search a pair $(\widetilde{m}, \widetilde{\sigma}^{\OS}) \neq (\widetilde{m}', \widetilde{\sigma}'{}^{\OS})$ in $\mathbb{L}^{\Sign}$ such that the first  three  \\  
~~~~~~elements of $\sigmaOS$ are the same.  i.e.,  $(\widetilde{\sfroot}, \widetilde{c}, \widetilde{\sigma}^{\DS}) = (\widetilde{\sfroot}', \widetilde{c}', \widetilde{\sigma}'{}^{\DS})$ \\
~~~If there is no such a pair, $\ForgeDS\leftarrow \tttrue$, return the game output $1$.\\
~~~~~~~~~~~$(\Final = \ttfalse \land \ReuseDS = \ttfalse \land \ForgeDS = \tttrue \land \CollCOM = \ttfalse)$\\
~~~$\ReuseDS \leftarrow \tttrue$.\\
~~~Parse $\widetilde{\sigma}^{\OS}$ as $(\sfroot^*, c^*, \sigmastrDS, \widetilde{\sfpath}, \widetilde{j}, \widetilde{r})$,
$\widetilde{\sigma}'{}^{\OS}$ as $(\sfroot^*, c^*, \sigmastrDS, \widetilde{\sfpath}', \widetilde{j}', \widetilde{r}')$. \\
~~~If $(\widetilde{m}, \widetilde{r})  \neq (\widetilde{m}', \widetilde{r}')$, $\CollCOM \leftarrow \tttrue$, return the game output $1$.\\
~~~~~~~~~~~$(\Final = \ttfalse \land \ReuseDS = \tttrue \land \ForgeDS = \ttfalse \land \CollCOM = \tttrue)$ \\
~~~Otherwise, return the game output $1$.\\
~~~~~~~~~~~$(\Final = \ttfalse \land \ReuseDS = \tttrue \land \ForgeDS = \ttfalse \land \CollCOM = \ttfalse)$\\
\hline
\end{tabular}
\caption{\small The base game $\sfGame^{\sfBasic}_{\onOS_{\Ours}, \A}$.}
\label{BasicSimOSour}
\end{figure}

$\sfGame^{\sfBasic}_{\onOS_{\Ours}, \A}$ simulates the game $\sfGame^{\SeqsEUFCMA}_{\onOS_{\Ours}, \A}$ by introducing flags (e.g., $\Final$, $\ReuseDS$) which are used for classifying forgery type and a table $\mathbb{T}$ which stores the computation of the signing oracle $\mathcal{O}^{\Sign}$.
More precisely, the flag $\Final$ represents that a forgery $(m^*, \sigmastrOS=(\sfroot^*, c^*, \sigmastrDS, \sfpath^*, \allowbreak j^*, r^*))$ is submitted in the final output ($\Final = \tttrue$) or $\mathcal{O}^{\Fin}$ ($\Final = \ttfalse$).
The flag $\ReuseDS$ represents that there is a pair $(\widetilde{m}, \widetilde{\sigma}^{\OS}) \neq (\widetilde{m}', \widetilde{\sigma}'{}^{\OS})$ in $\mathbb{L}^{\Sign}$ such that the first three elements of $\sigmaOS$ are the same.  i.e.,  $(\widetilde{\sfroot}, \widetilde{c}, \widetilde{\sigma}^{\DS}) = (\widetilde{\sfroot}', \widetilde{c}', \widetilde{\sigma}'{}^{\DS})$ holds.
We represent that such a pair exists as $\ReuseDS = \tttrue$.
 The table $\mathbb{T}$ stores a tuple $(i, M, \sfroot, c, \sigmaDS)$ where $(M, c)$ is an input for an $i$-th $\mathcal{O}^{\Sign}$ query, $(\sfroot, \sftree) \leftarrow \MerkleTree^{H}(M)$, $\sigmaDS \leftarrow \DSSign(\skDS, (\sfroot, c))$.
The counter $q^{\Sign}$ represents the number of outputs that $\A$ received from the $\mathcal{O}^{\Sign}$ oracle and $q^{\Fin}$ represent the number of submitted signatures from $\A$.

Now, we divide an adversary $\A$ into three types $\A_{1}, \A_{2}, \A_{3}$ according to states of flags $\ReuseDS$, $\ForgeDS$, and $\CollCOM$ when $\A$ wins the game $\sfGame^{\sfBasic}$.
\begin{itemize}
\item$\A_{1}$ wins the game with $\ForgeDS = \tttrue$.
\item$\A_{2}$ wins the game with $\CollCOM=\tttrue$.
\item$\A_{3}$ wins the game with $\ForgeDS = \ttfalse \land \CollCOM=\ttfalse$.
\end{itemize}

For adversaries $\A_{1}$, $\A_{2}$, and $\A_{3}$, we can construct a reduction for the security of $\DS$, $\COM$, and $H$ respectively.
Now, we give reductions for these adversaries.

\paragraph{\bf Reduction $\B^{\DS}$:}
A reduction $\B^{\DS}$ to the $\sEUFCMA$ security game of $\DS$ is obtained by modifying $\sfGame^{\sfBasic}_{\onOS_{\Ours}, \A}$ as follows.
Instead of running $\ppDS \leftarrow \DSSetup (1^{\lambda})$ and $(\vkDS, \skDS) \leftarrow \DSKeyGen(\ppDS)$, $\B^{\DS}$ uses $(\ppDS, \vkDS)$ given by the $\sEUFCMA$ security game of $\DS$.
For a signing query $(M, c)$ from $\A$, $\B^{\DS}$ query $(\sfroot, c)$ to the signing oracle of the $\sEUFCMA$ security game of $\DS$, obtains $\sigmaDS \leftarrow \DSSign(\skDS, \allowbreak (\sfroot, c))$, and returns $\sigmaDS$.
To simplify the discussion, we assume that $\A$ makes distinct $(M, c)$ to $\B^{\DS}$.
(If $\A$ makes the same $(M, c)$ more than once, $\B^{\DS}$ simply outputs return $\sigmaDS \leftarrow \DSSign(\skDS, \allowbreak (\sfroot, c))$ which was previously obtained by the signing oracle of the $\sEUFCMA$ security game where $\sfroot$ is computed from~$M$.) 

If $\B^{\DS}$ outputs $1$ with the condition where $\ForgeDS=\tttrue$, there is the forgery $(\widetilde{\sfroot}, \widetilde{c}, \widetilde{\sigma}^{\DS})$.
Since $\ForgeDS=\tttrue$ holds, $\ReuseDS=\ttfalse$ holds.
This fact implies that for $(m, \sigmaOS) \in \mathbb{L}^{\Sign}$, the first three elements $(\sfroot, c, \sigmaDS)$ of $\sigmaOS$ are all distinct in $\mathbb{L}^{\Sign}$ and valid signatures for $\DS$ (i.e., $\DSVerify(\vkDS,  \allowbreak \allowbreak (\sfroot, c), \sigmaDS)=1$).
Moreover, $\B^{\DS}$ makes $q^{\Sign}$ signing queries to signing oracle, $q^{\Sign} < q^{\Fin}$ holds where $q^{\Fin}$ is the number of entry in $\mathbb{L}^{\Sign}$. Hence, there is a forgery $((\widetilde{\sfroot}, \widetilde{c}), \widetilde{\sigma}^{\DS})$ of $\DS$.
By modifying $\sfGame^{\sfBasic}_{\onOS_{\Ours}, \A}$ to output this forgery $((\widetilde{\sfroot}, \widetilde{c}), \widetilde{\sigma}^{\DS})$, we can obtain~$\B^{\DS}$.

\paragraph{\bf Reduction $\B^{\COM}$:}
A reduction $\B^{\COM}$ to the strong computational binding property of $\COM$ is obtained by modifying $\sfGame^{\sfBasic}_{\onOS_{\Ours}, \A}$ as follows.
$\B^{\COM}$ uses $\ck$ given by the strong computational binding security game of $\COM$.

If $\sfGame_{\OS_{\Ours}, \A}$ outputs $1$ with the condition where $\CollCOM=\tttrue$, there is a collision $(\widetilde{m}, \widetilde{r})  \neq (\widetilde{m}', \widetilde{r}')$ such that $\COMCommit(\ck, \allowbreak \widetilde{m}; \widetilde{r}) = \COMCommit(\ck, \allowbreak \widetilde{m}'; \widetilde{r}')$ holds.
Since if $\CollCOM=\tttrue$ holds, $\ReuseDS=\tttrue$ holds in $\sfGame^{\sfBasic}_{\OS_{\Ours}, \A}$.
This fact implies that there is a pair $(\widetilde{m}, \widetilde{\sigma}^{\OS}=(\sfroot^*, c^*, \sigmastrDS, \allowbreak \widetilde{\sfpath},\allowbreak \widetilde{j}, \widetilde{r})) \neq (\widetilde{m}', \widetilde{\sigma}'{}^{\OS}= (\sfroot^*, c^*, \sigmastrDS, \widetilde{\sfpath}', \allowbreak \widetilde{j}', \widetilde{r}'))$.
Since $(\widetilde{m}, \widetilde{\sigma}^{\OS})$ and $(\widetilde{m}', \widetilde{\sigma}'{}^{\OS})$ are valid signatures, $(\widetilde{m}, \widetilde{r})  \neq (\widetilde{m}', \widetilde{r}')$ and $\COMCommit(\ck, \allowbreak \widetilde{m}; \widetilde{r}) = \COMCommit(\ck, \allowbreak \widetilde{m}'; \widetilde{r}')$ hold.
By modifying $\sfGame^{\sfBasic}_{\onOS_{\Ours}, \A}$ to output this collision $((\widetilde{m}, \widetilde{\sigma}^{\OS}), (\widetilde{m}', \allowbreak \widetilde{\sigma}'{}^{\OS}))$, we can obtain $\B^{\COM}$.

\paragraph{\bf Reduction $\B^{\Hash}$:}
We explain how to obtain a reduction $\B^{\Hash}$ to the collision resistance property from $\sfGame^{\sfBasic}_{\onOS_{\Ours}, \A}$.
If $\sfGame^{\sfBasic}_{\onOS_{\Ours}, \A}$ outputs $1$ with the condition where $\Final = \ttfalse \land \ReuseDS = \ttfalse \land \ForgeDS = \ttfalse$, a collision a hash function can be found.
Since $\Final = \ttfalse \land \ReuseDS = \ttfalse \land \ForgeDS = \ttfalse$ holds, then $(q^{\Sign}, *,  \sfroot^*, c^*, \sigmastrDS) \in \mathbb{T}$ holds.
Let $(M_{q^{\Sign}}, c_{q^{\Sign}})$ be an input for the $q^{\Sign}$-th $\mathcal{O}^{\Sign}$ query.
Then, by the computation of $\mathcal{O}^{\Sign}$ and table $\mathbb{T}$, $c^*=c_{q^{\Sign}}$, $(\sfroot^*, \sftree^*) = \MerkleTree^{H}(M_{q^{\Sign}})$, and $\DSVerify(\vkDS, (\sfroot^*, c^*), \sigmastrDS) = 1$ holds.
Since $m^* \notin M_{q^{\Sign}}$, a collision of a hash function $H$ can be computed by $(x, x') \leftarrow \Ext_{1}(\sftree^*, (m^*, \sfpath^*, i^*))$.
We modify $\sfGame^{\sfBasic}_{\onOS_{\Ours}, \A}$ to output this collision $(x, x')$ in this case.

If $\B^{\sfBasic}_{\onOS_{\Ours}, \A}$ outputs $1$ with the condition where $\ReuseDS = \tttrue \land \ForgeDS = \ttfalse \land \CollCOM = \ttfalse$ (regardless of the bool value $\Final $), a collision of a hash function can be also found.
Since $\ReuseDS = \ttfalse \land \CollCOM = \ttfalse$ holds, then there is a pair $(\widetilde{m}, \widetilde{\sigma}^{\OS}=(\sfroot^*, c^*, \sigmastrDS,  \allowbreak \widetilde{\sfpath},\allowbreak \widetilde{j}, \widetilde{r})) \neq (\widetilde{m}, \widetilde{\sigma}'{}^{\OS}=(\sfroot^*, c^*, \sigmastrDS,  \widetilde{\sfpath}', \allowbreak \widetilde{j}', \widetilde{r}))$ holds.
From this fact, we can see that $(\widetilde{\sfpath},\allowbreak \widetilde{j}) \neq (\widetilde{\sfpath}
', \allowbreak \widetilde{j}')$ holds.
If $j^* \neq \widetilde{j}$ holds, we can obtain a collision of a hash function H as $(x, x') \leftarrow \Ext_{1}(\sftree^*, (m^*, \sfpath^*, i^*)$. 
If $j^* = \widetilde{j}$ holds, then $\widetilde{\sfpath} = \widetilde{\sfpath}'$ holds and thus we can compute a collision of a hash function as $(x, x') \leftarrow \Ext_{2} (m, j^*,  \widetilde{\sfpath}, \widetilde{\sfpath}')$.
We modify $\sfGame^{\sfBasic}_{\onOS_{\Ours}, \A}$ to output this collision $(x, x')$ in these case.

By reduction algorithms $\B^{\DS}$, $\B^{\COM}$, and $\B^{\Hash}$ described above, we can bound the advantage $\Adv^{\SeqsEUFCMA}_{\onOS, \A}(1^{\lambda})$ as

\begin{equation*}
\begin{split}
\Adv^{\SeqsEUFCMA}_{\onOS, \A}({\lambda}) &= \Pr[\sfGame^{\SeqsEUFCMA}_{\onOS_{\Ours}, \A}(1^{\lambda}) \Rightarrow 1] = \Pr[\sfGame^{\sfBasic}_{\onOS_{\Ours}, \A}(1^{\lambda}) \Rightarrow 1]\\
&=  \Pr[\sfGame^{\sfBasic}_{\onOS_{\Ours}, \A}(1^{\lambda}) \Rightarrow 1 \land \ForgeDS = \tttrue]\\
&~~~~~~+  \Pr[\sfGame^{\sfBasic}_{\onOS_{\Ours}, \A}(1^{\lambda}) \Rightarrow 1 \land \CollCOM=\tttrue]\\
&~~~~~~+  \Pr[\sfGame^{\sfBasic}_{\onOS_{\Ours}, \A}(1^{\lambda}) \Rightarrow 1 \land \ForgeDS = \ttfalse \land \CollCOM=\ttfalse]\\
&\leq \Adv^{\sEUFCMA}_{\DS, \A_{1}}({\lambda})  + \Adv^{\sBind}_{\COM, \A_{2}}({\lambda})  +\Adv^{\Coll}_{\mathcal{H}, \A_{3}}({\lambda}).
\end{split}
\end{equation*}
By this fact, we can conclude Theorem \ref{OurconEUF}.
\qed
\end{proof}

\section{Conclusion}\label{SecConclude}
\paragraph{\bf Summary of Our Results.}
In this paper, we revisit the unforgeability security for a $1$-out-of-$n$ oblivious signature scheme and point out problems. By reflecting on these problems, we define the $\rmSeqsEUFCMA$ security.
We propose the improved generic construction of a $1$-out-of-$n$ oblivious signature scheme $\onOS_{\Ours}$.
Compared to the construction by Zhou et al. \cite{ZLH22}, our construction offers a smaller second message size. 
The sum of a second message size and a signature size is improved from $O(n)$ to $O(\log {n})$.

\paragraph{\bf Discussion of Our Security Model.}
We introduce the $\rmSeqsEUFCMA$ security in Definition \ref{Def_Seq-sEUF-CMA_Security}.
It is natural to consider a model that allows concurrent signing interactions. 
However, if we straightforwardly extend our security model to a concurrent setting, there is a trivial attack.

Let us consider the following adversary $\A$ that runs signing protocol executions twice concurrently.
$\A$ chooses two list $M_1=(m_{1,0}, \dots, m_{1, n-1})$ and $M_2=(m_{2,0},  \dots, m_{2, n-1})$ such that $M_1 \cap M_2 = \emptyset$ (i.e., there is no element $m$ such that $m \in M_1 \land m \in M_2$).
In the 1st interaction, $\A$ chooses $m_{1,0} \in M_{1}$, obtains a signature $\sigma_{1}$ on a message $m_{1, 0}$.
In the 2nd interaction $\A$ chooses $m_{2, 0} \in M_{2}$, obtains a signature $\sigma_{2}$ on a message $m_{2, 0}$.
$\A$ finishes the 1st interaction by outputting $(m_{2, 0}, \sigma_{2})$.
Since $m_{2, 0} \notin M_{1}$, $\A$ trivially wins the unforgeability game.
Due to this trivial attack, we cannot straightforwardly extend our security model to the concurrent signing interaction setting.

The drawback of our security model is that it seems somewhat complex and redundant.
It may be possible to define the unforgeability security more simply. 
For example, instead of the signature resubmission check, by introducing the algorithm $\Link$ which verifies a link between a signing interaction and a signature, we may simplify the our security model.
We leave more simple formalization of unforgeability security models as a future work.

\section*{Acknowledgement}
This work was supported by JST CREST Grant Number JPMJCR2113 and  JSPS KAKENHI Grant Number JP23K16841.
We also would like to thank anonymous referees for their constructive comments.

\bibliographystyle{abbrvurl}
\bibliography{ref}

\newpage
\setcounter{tocdepth}{2}
\tableofcontents

\end{document}